\documentclass[12pt]{article}

\usepackage{amsmath}
\usepackage{amssymb}
\usepackage{amsfonts}
\usepackage{latexsym}
\usepackage{color}

\catcode `\@=11 \@addtoreset{equation}{section}

\catcode `\@=12

\newcommand{\be}{\begin{equation}}
\newcommand{\en}{\end{equation}}
\newcommand{\bea}{\begin{eqnarray}}
\newcommand{\ena}{\end{eqnarray}}
\newcommand{\beano}{\begin{eqnarray*}}
\newcommand{\enano}{\end{eqnarray*}}
\newcommand{\bee}{\begin{enumerate}}
\newcommand{\ene}{\end{enumerate}}

\newcommand{\mc}{\mathcal}

\newcommand{\D}{{\mc D}}

\newcommand{\Sc}{{\cal S}}

\newcommand{\F}{{\cal F}}
\newcommand{\G}{{\cal G}}

\newcommand{\Lc}{{\cal L}}

\newcommand{\1}{1 \!\! 1}

\newcommand{\Hil}{\mc H}

\newtheorem{thm}{Theorem}

\newtheorem{prop}[thm]{Proposition}
\newtheorem{defn}[thm]{Definition}

\newenvironment{proof}{\noindent {\bf Proof --}}{\hfill$\square$ \vspace{3mm}\endtrivlist}

\catcode `\@=11 \@addtoreset{equation}{section}
\catcode `\@=12

\begin{document}

\thispagestyle{empty}

\vspace*{2cm}

\begin{center}
{\Large \bf From self-adjoint to non self-adjoint harmonic oscillators:\\ physical consequences and mathematical pitfalls}   \vspace{2cm}\\

{\large F. Bagarello}\\
  Dipartimento di Energia, Ingegneria dell'Informazione e Modelli Matematici,\\
Facolt\`a di Ingegneria, Universit\`a di Palermo,\\ I-90128  Palermo, Italy\\
e-mail: fabio.bagarello@unipa.it\\
home page: www.unipa.it/fabio.bagarello

\end{center}

\vspace*{2cm}

\begin{abstract}
\noindent Using as a prototype example the harmonic oscillator we show how losing self-adjointness of the hamiltonian $H$ changes drastically the related functional structure. In particular, we show that even a small deviation from strict self-adjointness of $H$ produces two deep consequences, not well understood in the literature: first of all, the original orthonormal basis of $H$ splits into two families of biorthogonal vectors. These two families are complete but, contrarily to what often claimed for similar systems, none of them is a basis for the Hilbert space $\Hil$. Secondly, the so-called metric operator is unbounded, as well as its inverse.

In the second part of the paper, after an extension of some previous results on the so-called $\D$ pseudo-bosons, we discuss some aspects of our extended harmonic oscillator from this different point of view.

\end{abstract}

\vspace{2cm}


\vfill


\newpage

\section{Introduction}

In recent years a larger and larger community of physicists and mathematicians started to be interested in some non self-adjoint operators having real eigenvalues, such as the celebrated hamiltonian $H=p^2+ix^3$, see \cite{petr} for a mathematically oriented treatment of $H$ and for many references. This original, and maybe restricted interest, was soon complemented by other related aspects, which include, for instance, {\em gain and loss} structures, see for instance \cite{circu}, as well as phase transitions, exceptional points and so on. We refer to \cite{ben,mosta,znorev} for some reviews on what is now called {\em pseudo-hermitian}, $PT$, or {\em crypto-hermitian} quantum mechanics.

In our opinion, most of the literature on these topics suffers from a sort of {\em original sin}: it is mainly written by physicists for physicists. This means that not much care about mathematical details is adopted, and this may have unpleasant consequences on the validity of the results deduced. For instance, in very many papers, the authors work with some non self-adjoint hamiltonian $h$, deducing the eigenstates of both $h$ and $h^\dagger$, and they simply claim that these two families are biorthogonal bases of the Hilbert space $\Hil$ where $h$ is defined. The aim of this paper is to convince the reader that such a procedure is, in fact, very dangerous, since already for extremely simple systems the two sets of eigenstates of $h$ and $h^\dagger$ are not bases at all! This is related, as discussed for instance in \cite{petr}, with the fact that the metric operator, or its inverse, is quite often an unbounded operator. And it is exactly for this reason that this kind of problems simply does not exist in finite dimensional Hilbert spaces, which are used quite often in the literature to produce examples of the general structure behind the models. However, if from one side finite-dimensional examples are usually easily handled, and for this reason they are very often proposed and analyzed, see \cite{zno} for instance, on the other side they hide completely the mathematical problems we have sketched above. Hence, trying to deduce a general structure out of only finite-dimensional systems can be a rather risky business.

In this paper we will consider an example based on the simplest system in quantum mechanics, the harmonic oscillator, and  some non self-adjoint extensions of it introduced by making use of the $\D$ pseudo-bosonic operators recently discussed in \cite{bagnewpb}. Among the other results, we will show that, for our particular system,  there exist sets of vectors in $\Hil$ which are not bases, but which still are complete (or total) in $\Hil$\footnote{Recall that, when orthonormality of a set of vectors $\F$ is lost, the two concepts ($\F$ being a basis or $\F$ being complete) are different: the first implies the second, but the opposite is false, \cite{heil}.}. They are eigenstates of the hamiltonian of the system and of its adjoint.

This article is organized as follows: in the next section we review and extend the definition and few useful results on $\D$ pseudo-bosons ($\D$-PBs). In Section III we introduce the harmonic oscillator and some non self-adjoint extensions. For these extended systems we deduce that the related biorthogonal sets of eigenstates are not Riesz bases. As a matter of fact, we prove that they are not even bases. In Section IV we show how to use $\D$-PBs to deal with these extended oscillators, and we prove that the metric operator is unbounded, as well as its inverse. Section V contains our conclusions.

\section{$\D$ pseudo-bosons}\label{sectII}

We briefly review here few facts and definitions on $\D$-PBs. More details can be found in \cite{bagnewpb}.

Let $\Hil$ be a given Hilbert space with scalar product $\left<.,.\right>$ and related norm $\|.\|$. Let further $a$ and $b$ be two operators
on $\Hil$, with domains $D(a)$ and $D(b)$ respectively, $a^\dagger$ and $b^\dagger$ their adjoint, and let $\D$ be a dense subspace of $\Hil$
such that $a^\sharp\D\subseteq\D$ and $b^\sharp\D\subseteq\D$, where $x^\sharp$ is $x$ or $x^\dagger$. Incidentally, it may be worth noticing
that we are not requiring here that $\D$ coincides with, e.g. $D(a)$ or $D(b)$. Nevertheless, for obvious reasons, $\D\subseteq D(a^\sharp)$
and $\D\subseteq D(b^\sharp)$.

\begin{defn}\label{def21}
The operators $(a,b)$ are $\D$-pseudo bosonic ($\D$-pb) if, for all $f\in\D$, we have
\be
a\,b\,f-b\,a\,f=f.
\label{31}\en
\end{defn}
 Sometimes, to simplify the notation, instead of (\ref{31}) we will simply write $[a,b]=\1$, having in mind that both sides of this equation
have to act on $f\in\D$.

\vspace{2mm}

Our  working assumptions are the following:

\vspace{2mm}

{\bf Assumption $\D$-pb 1.--}  there exists a non-zero $\varphi_{ 0}\in\D$ such that $a\,\varphi_{ 0}=0$.

\vspace{1mm}

{\bf Assumption $\D$-pb 2.--}  there exists a non-zero $\Psi_{ 0}\in\D$ such that $b^\dagger\,\Psi_{ 0}=0$.

\vspace{2mm}

Then, if $(a,b)$ satisfy Definition \ref{def21}, it is obvious that $\varphi_0\in D^\infty(b):=\cap_{k\geq0}D(b^k)$ and that $\Psi_0\in D^\infty(a^\dagger)$, so
that the vectors \be \varphi_n:=\frac{1}{\sqrt{n!}}\,b^n\varphi_0,\qquad \Psi_n:=\frac{1}{\sqrt{n!}}\,{a^\dagger}^n\Psi_0, \label{32}\en
$n\geq0$, can be defined and they all belong to $\D$ and, as a consequence, to the domains of $a^\sharp$, $b^\sharp$ and $N^\sharp$, where $N=ba$. We introduce, as in \cite{bagnewpb}, $\F_\Psi=\{\Psi_{ n}, \,n\geq0\}$ and
$\F_\varphi=\{\varphi_{ n}, \,n\geq0\}$.

It is now simple to deduce the following lowering and raising relations:
\be
\left\{
    \begin{array}{ll}
b\,\varphi_n=\sqrt{n+1}\varphi_{n+1}, \qquad\qquad\quad\,\, n\geq 0,\\
a\,\varphi_0=0,\quad a\varphi_n=\sqrt{n}\,\varphi_{n-1}, \qquad\,\, n\geq 1,\\
a^\dagger\Psi_n=\sqrt{n+1}\Psi_{n+1}, \qquad\qquad\quad\, n\geq 0,\\
b^\dagger\Psi_0=0,\quad b^\dagger\Psi_n=\sqrt{n}\,\Psi_{n-1}, \qquad n\geq 1,\\
       \end{array}
        \right.
\label{33}\en as well as the following eigenvalue equations: $N\varphi_n=n\varphi_n$ and  $N^\dagger\Psi_n=n\Psi_n$, $n\geq0$. As a consequence
of these  equations,  choosing the normalization of $\varphi_0$ and $\Psi_0$ in such a way $\left<\varphi_0,\Psi_0\right>=1$, we deduce that
\be \left<\varphi_n,\Psi_m\right>=\delta_{n,m}, \label{34}\en
 for all $n, m\geq0$. The third assumption we introduced in \cite{bagnewpb} is the following:

\vspace{2mm}

{\bf Assumption $\D$-pb 3.--}  $\F_\varphi$ is a basis for $\Hil$.

\vspace{1mm}

This is equivalent to the request that $\F_\Psi$ is a basis for $\Hil$, \cite{bagnewpb}. In particular, if $\F_\varphi$ and $\F_\Psi$ are Riesz basis for $\Hil$, we have  called our $\D$-PBs {\em regular}.

\vspace{2mm}

{\bf Remark:--}  We recall once again that requiring that $\F_\varphi$ is a basis is much more, for non o.n. sets, than requiring
that $\F_\varphi$ is just complete. Counterexamples can be found in \cite{bagnewpb,heil}.

\vspace{2mm}

In \cite{bagnewpb} we have introduced a weaker version of Assumption $\D$-pb 3, useful for physical applications: for that, let $\G$ be a suitable dense subspace of $\Hil$. Two biorthogonal sets $\F_\eta=\{\eta_n\in\G,\,g\geq0\}$ and $\F_\Phi=\{\Phi_n\in\G,\,g\geq0\}$ have been called {\em $\G$-quasi bases} if, for all $f, g\in \G$, the following holds:
\be
\left<f,g\right>=\sum_{n\geq0}\left<f,\eta_n\right>\left<\Phi_n,g\right>=\sum_{n\geq0}\left<f,\Phi_n\right>\left<\eta_n,g\right>.
\label{35}
\en
Is is clear that, while Assumption $\D$-pb 3 implies (\ref{35}), the reverse is false. However, if $\F_\eta$ and $\F_\Phi$ satisfy (\ref{35}), we still have some (weak) form of resolution of the identity.  Now Assumption $\D$-pb 3 is replaced by the following:

\vspace{2mm}

{\bf Assumption $\D$-pbw 3.--}  For some subspace $\G$ dense in $\Hil$, $\F_\varphi$ and $\F_\Psi$ are $\G$-quasi bases.

\subsection{$\Theta$-conjugate operators for $\D$-quasi bases}

In this section we slightly refine the structure. Notice that, with respect to what done in \cite{bagnewpb}, we will here assume that   Assumption $\D$-pb 1,  $\D$-pb 2, and  $\D$-pbw 3 are satisfied, with $\G\equiv\D$. In other words, we will not assume  $\D$-pb 3, since this Assumption, even if it is very often taken for granted in the physical literature on non self-adjoint hamiltonians, is not  satisfied even in our simple extended harmonic oscillator, see Section III.

Let us consider a self-adjoint, invertible, operator $\Theta$, which leaves, together with $\Theta^{-1}$, $\D$ invariant: $\Theta\D\subseteq\D$, $\Theta^{-1}\D\subseteq\D$. Then, as in \cite{bagnewpb}, we say that $(a,b^\dagger)$ are $\Theta-$conjugate if $af=\Theta^{-1}b^\dagger\,\Theta\,f$, for all $f\in\D$.

Replacing $\D$-pb 3 with $\D$-pbw 3 does not influence the fact that, for instance, $(a,b^\dagger)$ are $\Theta-$conjugate if and only if  $(b,a^\dagger)$ are $\Theta-$conjugate. This is because, see \cite{bagnewpb}, the proof of this equivalence does not make any use of the nature of the sets $\F_\varphi$ and $\F_\Psi$. On the other hand, this is used in \cite{bagnewpb} to prove that
 $(a,b^\dagger)$ are $\Theta-$conjugate if and only if $\Psi_n=\Theta\varphi_n$, for all $n\geq0$,
so one may imagine that this is no longer true in our new hypotheses. However, this is not so, and we can in fact prove the following
\begin{prop}\label{prop1}
Assume that $\F_\varphi$ and $\F_\Psi$ are $\D$-quasi bases for $\Hil$. Then the operators $(a,b^\dagger)$ are $\Theta-$conjugate if and only if $\Psi_n=\Theta\varphi_n$, for all $n\geq0$.
\end{prop}
\begin{proof}
The proof is quite similar to that in \cite{bagnewpb}, so that here we will stress only the differences. Notice that, as in \cite{bagnewpb}, we will use the following normalization: $\left<\varphi_0,\Theta\varphi_0\right>=1$.

Let us first assume that $(a,b^\dagger)$ are $\Theta-$conjugate. Then, as in \cite{bagnewpb}, we can check that $\F_{\tilde\varphi}=\{\tilde\varphi_n:=\Theta\varphi_n,\,n\geq0\}$ is biorthogonal to $\F_\varphi$. Now, it is easy to see that both $\F_\varphi$ and $\F_\Psi$ are complete in $\D$. In other words, if $f\in\D$ is orthogonal to all the $\varphi_n$'s or to all the $\Psi_n$'s, then $f=0$. Hence, since for all fixed $k$
$$
\left<\tilde\varphi_k-\Psi_k,\varphi_n\right>=\left<\tilde\varphi_k,\varphi_n\right>-\left<\Psi_k,\varphi_n\right>=\delta_{k,n}-\delta_{k,n}=0,
$$
 $\forall\, n\geq0$, and since $\tilde\varphi_k-\Psi_k$ belongs to $\D$, we conclude that $\tilde\varphi_k=\Psi_k$ for each $k$. Hence $\Psi_k=\Theta\varphi_k$.

\vspace{2mm}

Let us now assume that $\Psi_n=\Theta\varphi_n$, for all $n\geq0$. Then, as in \cite{bagnewpb}, we deduce that, taking $f$  in $\D$,
$$
\left<\left(\Theta\,a\,\Theta^{-1}-b^\dagger\right)f,\varphi_{n}\right>=\left<f,\left(\Theta^{-1}\,a^\dagger\,\Theta-b\right)\varphi_{n}\right>=0,
$$
for all $n\geq0$. Hence, since $\F_\varphi$ is complete in $\D$, we conclude that $\left(\Theta\,a\,\Theta^{-1}-b^\dagger\right)f=0$ for each $f\in\D$, so that $(b^\dagger,a)$ are $\Theta^{-1}$-conjugate, which in turns implies our statement.

\end{proof}

The positivity of $\Theta$, proved in \cite{bagnewpb} under stronger assumptions, can also be deduced in the present settings, again with a really minor difference. We have

\begin{prop}\label{prop2}
If $(a,b^\dagger)$ are $\Theta-$conjugate, then $\left<f,\Theta f\right>>0$ for all non zero $f\in \D$.
\end{prop}

\begin{proof}

Since both $f$ and $\Theta f$ belong to $\D$, and since  $\F_\varphi$ and $\F_\Psi$ are $\D$-quasi bases for $\Hil$, the following expansion holds

$$
\left<f,\Theta f\right>=\sum_n\left<f,\Psi_n\right>\left<\varphi_n,\Theta f\right>=\sum_n\left<f,\Psi_n\right>\left<\Theta\varphi_n, f\right>=\sum_n\left<f,\Psi_n\right>\left<\Psi_n, f\right>=\sum_n|\left<f,\Psi_n\right>|^2,
$$
which is surely strictly positive if $f\neq0$, due to the fact that $\F_\Psi$ is complete in $\D$.

\end{proof}

Notice that, in \cite{bagnewpb}, $f$ was any vector in $D(\Theta)$, which could be larger than $\D$. Here we need to restrict to $\D$. Notice also that, not surprisingly, we can deduce that
$Nf=\Theta^{-1}N^\dagger\Theta f$, for all $f\in \D$.

\vspace{2mm}

We end this section with a final remark: it is clear by the definition that $(a,b^\dagger)$ are $\Theta-$conjugate if and only if $(a,b^\dagger)$ are $\Theta_k-$conjugate, where $\Theta_k:=k\,\Theta$, for all possible choices of non zero real $k$. Here, reality of $k$ is needed to ensure that $\Theta_k$ is self-adjoint. Notice that $k$ could also be negative, in principle. This could seem to be in contradiction with Propositions \ref{prop1} and \ref{prop2}. In fact, this is not so, since these results are deduced under the requirement that $\left<\varphi_0,\Theta\varphi_0\right>=1$, which of course fixes the value of the constant $k$.

\section{The harmonic oscillator: loosing self-adjointness}

This section is devoted to a detailed analysis of the {\em shifted} harmonic oscillator, and of some of its possible non self-adjoint extensions. It may be worth stressing that in the literature several such extensions exist, in one or more spatial dimensions, see \cite{ben2} and references therein for some examples. We should also mention that part of the results we will discuss in  this section are somehow related to a similar model we have recently introduced in \cite{bagpb2}, but with a rather different perspective.

The original ingredients of our analysis are the self-adjoint position and momentum operators $x$ and $p=-i\frac{d}{dx}$, satisfying $[x,p]=i\1$, and the standard annihilation and creation operators $a=\frac{1}{\sqrt{2}}(x+ip)$ and $a^\dagger=\frac{1}{\sqrt{2}}(x-ip)$ constructed out of $x$ and $p$, which obey $[a,a^\dagger]=\1$. Here $\1$ is the identity operator on $\Hil=\Lc^2(\Bbb R)$.

Let now fix $k\in \Bbb R$, $\alpha, \beta\in \Bbb C$, and let us introduce the operators
\be
c:=a+k,\qquad A=a+\alpha,\qquad B=a^\dagger+\overline{\beta},
\label{41}\en
as well as their adjoints $c^\dagger=a^\dagger+k$, $A^\dagger=a^\dagger+\overline{\alpha}$ and $B^\dagger=a+\beta$. The commutation rules
\be
[c,c^\dagger]=[c,A^\dagger]=[c,B]=[A,c^\dagger]=[A,A^\dagger]=[A,B]=[B,c]=[B,A]=[B,B^\dagger]=\1,
\label{42}\en
suggest that $(c,c^\dagger)$, $(A,A^\dagger)$ and $(B,B^\dagger)$ are bosonic operators for all choices of $k$, $\alpha$ and $\beta$, while, for instance, the pairs $(A,c^\dagger)$, $(A,B)$ or $(B,c)$ are (at least formally), pseudo-bosonic. Therefore, we could introduce several self-adjoint, and non self-adjoint, number operators, like $\hat n=c^\dagger c$, $A^\dagger A$, $BB^\dagger$, $N=BA$ and $N^\dagger=A^\dagger B^\dagger$, and so on. The main result of our analysis will allow us to conclude that, while the eigenstates of $\hat n$ produce an orthonormal (o.n.) basis for $\Hil$, the eigenstates of $N$ (or those of $N^\dagger$) are not even a basis.
For that we will make use of the unitary displacement operator
\be
D(z)=e^{\overline{z}\,a-za^\dagger}=e^{-iz_iz_r}e^{-i\sqrt{2}z_ix}e^{i\sqrt{2}z_rp},
\label{43}\en
where $z=z_r+iz_i$. The role of $D(z)$ is important, since we can write $c=D(k)aD^{-1}(k)$, $A=D(\alpha)aD^{-1}(\alpha)$ and $B=D(\beta)a^\dagger D^{-1}(\beta)$. Hence, calling $\hat n_0:=a^\dagger a$, it is clear that $\hat n=D(k)\hat n_0 D^{-1}(k)$, $N=D(\beta)a^\dagger D^{-1}(\beta)D(\alpha)aD^{-1}(\alpha)$ and $N^\dagger=D(\alpha)a^\dagger D^{-1}(\alpha)D(\beta)aD^{-1}(\beta)$. In particular, we see that $N=N^\dagger$ if $\alpha=\beta$, but not in general.

\subsection{The self-adjoint shifted harmonic oscillator}

Let us first recall that, in coordinate representation, the normalized vacuum of $a$, $a\,e_0(x)=0$, is $e_0(x)=\frac{1}{\pi^{1/4}}\,e^{-x^2/2}$, and that the other eigenstates of $\hat n_0$ can be written as
$$
e_n(x)=\frac{1}{\sqrt{n!}}(a^\dagger)^n\,e_0(x)=\frac{1}{\sqrt{2^nn!\sqrt{\pi}}}\,H_n(x)\,e^{-x^2/2},
$$
where $H_n(x)$ is the $n$-th Hermite polynomial. The eigenstates $\Phi_n(x)$ of $\hat n$ can be easily deduced, both with a direct computation, or from the previous ones, simply because $c\,\Phi_0(x)=0$ produces, choosing properly the normalization, $\Phi_0(x)=D(k)e_0(x)$. This relation can be extended to the other functions of the two sets, $\F_e=\{e_n(x), n\geq0\}$ and $\F_\Phi=\{\Phi_n(x), n\geq0\}$. Indeed we find
\be
\Phi_n(x)=D(k)e_n(x)=e^{i\sqrt{2}kp}e_n(x)=e_n(x+\sqrt{2}k),
\label{44}\en
for all $n\geq0$ and for all real choices of $k$. In particular, the last equality follows from (\ref{43}). Since $\hat n \Phi_n(x)=n\Phi_n(x)$, we conclude that the eigenstates of the self-adjoint operator $\hat n$ are just the translated version of those of $\hat n_0$. They are clearly o.n., complete\footnote{We have previously introduces the notion of {\em completeness in $\D$}. To simplify the notation, here and in the following, we will simply say {\em complete} to mean complete in $\Hil$.}, and span all the Hilbert space. These properties could be easily deduced from the fact that $D(k)$ is unitary\footnote{We include here the proofs of the statements since they are important to show that, as soon as we move from unitary (and therefore bounded) to unbounded operators, many of the apparently obvious properties of the systems are simply lost.}. In fact, for instance, if $f\in \Hil$ is orthogonal to all the $\Phi_n(x)$'s, then for all $n\geq0$ we have
$$
0=\left<f, \Phi_n\right>=\left<f, D(k)e_n\right>=\left<D^{-1}(k)f, e_n\right> \qquad \Rightarrow\qquad D^{-1}(k)f=0,
$$
since $\F_e$ is complete. Hence, $f=0$, which implies that $\F_\Phi$ is complete as well. To prove that $\F_\Phi$ is also a basis for $\Hil$, we use the fact that, $\forall g\in\Hil$, $g=\sum_{n\geq0}\left<e_n,g\right>e_n$. Then we have, using the fact that $D(k)$ is continuous and the relation between $e_n$ and $\Phi_n$,
$$
f=D(k)\left(D^{-1}(k)f\right)=D(k)\left(\sum_{n\geq0}\left<e_n,D^{-1}(k)f\right>e_n\right)=$$
$$=\sum_{n\geq0}\left<D(k)e_n,f\right>D(k)e_n=\sum_{n\geq0}\left<\Phi_n,f\right>\Phi_n,
$$
 for all $f\in\Hil$. It is important to stress that what we have done here is only possible since $D(k)$ and $D^{-1}(k)$ are bounded. Otherwise, for instance, in the proof of the completeness of $\F_\Phi$ we should have taken $f$ in the domain of $D^{-1}(k)$, and this would not allow us to conclude. Also, in the previous equation, $D(k)$ could not be {\em taken inside} the infinite sum on $n$, since, in this case, there is no guarantee that the series $\sum_{n=0}^N\left<D(k)e_n,f\right>D(k)e_n$ converge.

\subsection{The non self-adjoint shifted harmonic oscillator}

Among the possible generalizations of the number operators $\hat n_0$ and $\hat n$, we could consider $A^\dagger A$ or $BB^\dagger$. However, since these appear both self-adjoint,  not many differences are expected with respect to what we have deduced previously. For instance, if we act with powers of $A^\dagger$ on the vacuum of $A$, $\varphi_0$, again we get an o.n. basis for $\Hil$, whose $n$-th vector can be written as $D(\alpha)e_n(x)$, and which satisfies the eigenvalue equation $(A^\dagger A)\left(D(\alpha)e_n(x)\right)=n\left(D(\alpha)e_n(x)\right)$. For our purposes it is more interesting to act on $\varphi_0$ with powers of $B$, and this is what we will do in some details here.

First, let us observe that, choosing a suitable normalization, $A\varphi_0=0$ if $\varphi_0=D(\alpha)e_0$. Then we have, see Section II,
\be
\varphi_n=\frac{1}{\sqrt{n!}}\,B^n\varphi_0=\frac{1}{\sqrt{n!}}\,D(\beta)(a^\dagger)^n
D^{-1}(\beta)D(\alpha)\,e_0=V(\alpha,\beta)\,e_n,
\label{45}\en
where we have introduced the operator $V(\alpha,\beta)=e^{\frac{1}{2}\alpha(\overline{\beta}-\overline{\alpha})}e^{\overline{\beta}a-\alpha a^\dagger}$. It is important to stress that $V(\alpha,\alpha)=D(\alpha)$, which means that, for some values of, say, $\beta$ the operator $V(\alpha,\beta)$ is unitary and, therefore, bounded. However, if $\alpha\neq\beta$, we will see later that $V(\alpha,\beta)$ is unbounded. Due to (\ref{45}), each $e_n$ belongs to its domain, which is dense in $\Hil$ since it contains the set of all the linear combinations of the $e_n$'s.

As discussed in Section II, the biorthogonal set $\Psi_n$ is constructed by first looking for the vacuum of $B^\dagger$: $B^\dagger\Psi_0=0$. This is satisfied if $a(D^{-1}(\beta)e_0)=0$, and then we deduce that $\Psi_0=\mu(\alpha,\beta)D(\beta)e_0$, where $\mu(\alpha,\beta):=e^{\frac{1}{2}(|\alpha|^2+|\beta|^2)-\beta\,\overline{\alpha}}$ is a suitable normalization, see below, needed ensure that $\left<\varphi_0,\Psi_0\right>=1$. If we act on $\Psi_0$ with powers of $B$ we construct an orthogonal basis for $\Hil$ of eigenstates of $BB^\dagger$. Suppose instead that we are interested in finding the eigenstates of $N^\dagger=A^\dagger B^\dagger$. Then, following Section II, we construct the new vectors
\be
\Psi_n=\frac{1}{\sqrt{n!}}\,(A^\dagger)^n\Psi_0=\frac{\mu(\alpha,\beta)}{\sqrt{n!}}\,D(\alpha)(a^\dagger)^n
D^{-1}(\alpha)D(\beta)\,e_0=\mu(\alpha,\beta)V(\beta,\alpha)\,e_n,
\label{46}\en
where we stress that $V(\beta,\alpha)$ appears rather than $V(\alpha,\beta)$, see (\ref{45}). A simple computation shows that, with our previous choice of normalization for $\varphi_0$ and $\Psi_0$,
$$
\left<\varphi_0,\Psi_0\right>=\left<V(\alpha,\beta)e_0,\mu(\alpha,\beta)V(\beta,\alpha)e_0\right>=\mu(\alpha,\beta)
e^{-\frac{1}{2}(|\alpha|^2+|\beta|^2)+\beta\,\overline{\alpha}}=1,
$$
so that $\left<\varphi_n,\Psi_m\right>=\delta_{n,m}$, for all $n,m\geq0$. The vectors of the sets $\F_\varphi$ and $\F_\Psi$ are respectively eigenstates of $N$ and $N^\dagger$, and they are biorthogonal. This is what quite often, in the literature is assumed to be sufficient to claim that $\F_\varphi$ and $\F_\Psi$ are indeed bases for $\Hil$. We will now prove that, on the contrary, neither $\F_\varphi$ nor $\F_\Psi$ can be basis for $\Hil$. Our argument extends that originally given in \cite{bagpb2}.

We start proving that $\lim_{n\rightarrow\infty}\|\varphi_n\|=\infty$. In fact, with a little algebra, we have
$$
\|\varphi_n\|^2=\left<V(\alpha,\beta)\,e_n,V(\alpha,\beta)\,e_n\right>=\|e^{(\overline{\beta}-\overline{\alpha})a}e_n\|^2\geq 1+|\overline{\beta}-\overline{\alpha}|^2n=1+|{\beta}-{\alpha}|^2n,
$$
which clearly diverges with $n$ diverging. The inequality above follows from the fact that, $\forall\gamma\in \Bbb C$,
$$
e^{\gamma a}\,e_n=\sum_{k=0}^n\,\frac{(\gamma a)^k}{k!}\,e_n=e_n+\gamma\sqrt{n}e_{n-1}+\cdots+\frac{\gamma^n}{\sqrt{n!}}\,e_0,
$$
and from the orthogonality of the different $e_k$'s. In a similar way we can also prove that $\lim_{n\rightarrow\infty}\|\Psi_n\|=\infty$: the different choice of normalization, in fact, does not produce any serious difference, at least under this aspect.

{\bf Remark:--} An essential point to stress is that the divergence of both $\|\varphi_n\|$ and $\|\Psi_n\|$ is only true if $\alpha\neq\beta$, which is exactly what we expect since, if $\alpha=\beta$, $\F_\varphi$ and $\F_\Psi$ both coincide with $\F_\Phi$, with $k=\alpha=\beta$, whose vectors are, in particular, normalized.

\vspace{2mm}

A consequence of these results is that $V(\alpha,\beta)$, if $\alpha\neq\beta$, is necessarily unbounded. The reason is simple: suppose this would not be so, and let $M$ be the (finite) norm of  $V(\alpha,\beta)$. Hence, since $\|\varphi_n\|=\|V(\alpha,\beta)e_n\|\leq M$, we would get a contradiction. An immediate consequence of this fact is that neither $\F_\varphi$ nor $\F_\Psi$ can be Riesz basis for $\Hil$, because a Riesz basis is the image of an o.n. basis via a bounded operator, with bounded inverse. However, this would not prevent $\F_\varphi$ or $\F_\Psi$, or both, to be bases. Nevertheless, we will now show that, for  $\alpha\neq\beta$, this is also impossible.

In fact, let us assume for the moment that $\F_\varphi$ is a basis for $\Hil$. Hence each $f\in\Hil$ can be written as $f=\sum_{n=0}^\infty\left<\Psi_n,f\right>\varphi_n=\sum_{n=0}^\infty P_n(f)$, where $P_n(f):=\left<\Psi_n,f\right>\varphi_n$. Since $\|P_n\|=\|\varphi_n\|\|\Psi_n\|\rightarrow\infty$, $\sup_n\|P_n\|=\infty$ and, as a consequence, the above expansion cannot converge for all vectors $f$. Hence, $\F_\varphi$ cannot be a basis for $\Hil$. In a similar way we can conclude that $\F_\Psi$ cannot be a basis for $\Hil$. Nevertheless, we will see in the next section that they still produce some useful weak form of the resolution of the identity.

Summarizing we have that, in our very simple model, {\em the biorthogonal sets of eigenstates of $N$ and $N^\dagger$ are not bases for $\Hil$.} This suggests that most of the claims which one can find in the physical literature on this subject, where the non self-adjoint hamiltonians are by far more complicated than the number operators $N$ and $N^\dagger$ considered here, are wrong or, at least, need to be justified in more details.

We should also stress that, even if they are not bases, both $\F_\varphi$ and $\F_\Psi$ are complete in $\Lc^2(\Bbb R)$. This is because both these sets are made of polynomials times a shifted gaussian. More in details, for instance, $\varphi_n(x)=p_n(x)e^{-(x-\gamma)^2/2}$, where $p_n(x)$ is a polynomial of degree $n$ and $\gamma$ is some fixed shift parameter. Then $\F_\varphi$ is complete, \cite{kolfom}. Our results show explicitly that, when orthonormality is lost, a complete set needs not to be a basis!

We end this section with some similarity relations, which can be proved explicitly, by using simple formulas for $a$ and $a^\dagger$:
$$
V^{-1}(\alpha,\beta)NV(\alpha,\beta)= V^{-1}(\beta,\alpha)N^\dagger V(\beta,\alpha)=\hat n_0,\qquad T^{-1}(\alpha,\beta)N T(\alpha,\beta)=N^\dagger,
$$
where $$T(\alpha,\beta)=V(\alpha,\beta)V^{-1}(\beta,\alpha)=e^{\frac{1}{2}(\alpha\overline{\beta}-\beta\overline{\alpha}+2|\beta|^2-2|\alpha|^2)}
e^{a(\overline{\beta}-\overline{\alpha})+a^\dagger(\beta-\alpha)}.$$ Notice that all these identities cannot be defined in all of $\Hil$, since the operators involved are unbounded.

\section{A $\D$-pb view to the non self-adjoint harmonic oscillator}

Among all the possible choices of formal pseudo-bosonic operators, we will here only consider the pair $(A,B)$, as in Section III.2. We have, see (\ref{41}): $A=\frac{1}{\sqrt{2}}\left(x+\frac{d}{dx}+\sqrt{2}\,\alpha\right)$ and $B=\frac{1}{\sqrt{2}}\left(x-\frac{d}{dx}+\sqrt{2}\,\overline{\beta}\right)$. The vacua of $A$ and $B^\dagger$ are then $\varphi_0(x)=N_\varphi\exp\{-\left(\frac{x^2}{2}+\sqrt{2}\,\alpha\,x\right)\}$ and $\Psi_0(x)=N_\Psi\exp\{-\left(\frac{x^2}{2}+\sqrt{2}\,\beta\,x\right)\}$. Here $N_\varphi$ and $N_\Psi$ must satisfy the following equality:
$$
\overline{N_\varphi}\,N_\Psi=\frac{1}{\sqrt{\pi}}\,e^{-(\beta+\overline{\alpha})^2/2},
$$
which ensure that $\left<\varphi_0,\Psi_0\right>=1$. Of course, $N_\varphi$ and $N_\Psi$ can be related to $\mu(\alpha,\beta)$ introduced
before, but this is not relevant for us. Both vacua belong to the set $\D=\{f(x)\in \Sc(\Bbb R):\,e^{kx}f(x)\in\Sc(\Bbb R), \, \forall k\in
\Bbb R\}$. This set is dense in $\Lc^2(\Bbb R)$, since contains $D(\Bbb R)$, and is stable under the action of $A^\sharp$ and $B^\sharp$.
Assumptions $\D$-pb 1 and $\D$-pb 2 are satisfied. As for Assumption $\D$-pb 3, our previous results show that this does not hold. However, it
is possible to show that Assumption $\D$-pbw 3 is satisfied, with $\G\equiv\D$. Indeed let us take $f,g\in \D$. Then, since $\F_e$ is an o.n.
basis and since $\D\subseteq D(V^\dagger(\alpha,\beta))\cap D(V^{-1}(\alpha,\beta))$, we have {\be
\left<f,g\right>=\left<V^\dagger(\alpha,\beta)f,V^{-1}(\alpha,\beta)g\right>=\sum_n
\left<V^\dagger(\alpha,\beta)f,e_n\right>\left<e_n,V^{-1}(\alpha,\beta)g\right>= \sum_n \left<f,\varphi_n\right>\left<\Psi_n,g\right>,
\label{48meno}\en since $\varphi_n=V(\alpha,\beta)e_n$ and, with few computations, we also find $\Psi_n=(V^\dagger(\alpha,\beta))^{-1}e_n$. }
Equation (\ref{48meno}) shows that, as required by Assumption $\D$-pbw 3, $\F_\varphi$ and $\F_\Psi$  produce a weak form of the resolution of
the identity.

\vspace{2mm}

Let us now put together $\varphi_n=V(\alpha,\beta)\,e_n$ and $\Psi_n=\mu(\alpha,\beta)V(\beta,\alpha)\,e_n$. Then we deduce that \be\Psi_n=\Theta(\alpha,\beta)\varphi_n,\label{47}\en where
\be
\Theta(\alpha,\beta)=\mu(\alpha,\beta)V(\beta,\alpha)V^{-1}(\alpha,\beta)=e^{-\frac{1}{2}|\alpha+\beta|^2+2|\alpha|^2}
e^{a(\overline{\alpha}-\overline{\beta})+a^\dagger(\alpha-\beta)},
\label{48}\en
which can also be written as $\Theta(\alpha,\beta)=e^{|\alpha|^2-|\beta|^2}e^{a(\overline{\alpha}-\overline{\beta})} e^{a^\dagger(\alpha-\beta)}$. We see that $\Theta(\alpha,\beta)$ is self-adjoint, leaves $\D$ invariant together with its inverse, and that $\Theta(\alpha,\alpha)=\1$. Moreover, a simple computation shows that $\Theta^{-1}(\alpha,\beta)B^\dagger \Theta(\alpha,\beta)f=Af$, for all $f\in\D$, so that $(A,B^\dagger)$ are $\Theta(\alpha,\beta)$-conjugate. This, in view of (\ref{47}), is exactly the content of Proposition \ref{prop1}. A simple consequence of this is, for instance, that $\Theta^{-1}(\alpha,\beta)N^\dagger \Theta(\alpha,\beta)f=Nf$, for all $f\in\D$. Finally, for each non zero $f\in\D$, we deduce that $\left<f,\Theta(\alpha,\beta)f\right>=e^{|\alpha|^2-|\beta|^2}\|e^{a^\dagger(\alpha-\beta)}f\|^2$, which is strictly positive. Hence we recover, for our simple model, the general structure and results discussed in Section II.

We should stress that $\Theta(\alpha,\beta)$ is what in the literature is usually called {\em the metric operator}, and an estimate of the kind already used above for $V(\alpha,\beta)$ allows us to conclude that both $\Theta(\alpha,\beta)$ and its inverse are unbounded operators, at least if $\alpha\neq\beta$. Again, this result contradicts what is usually assumed in the literature, i.e. that the metric operator and its inverse are (at least one of them) bounded. This seems to be not so automatic, and need to be checked even in very simple systems.

\section{Conclusions}

After a preliminary section on $\D$-PBs, we have considered some manifestly non self-adjoint extensions of the harmonic oscillator producing two number-like operators related by the adjoint operation in $\Lc^2(\Bbb R)$, $N$ and $N^\dagger$. For these operators we have deduced the related eigenstates, and we have proved that they form two biorthogonal, complete families of $\Hil$ which {\bf are not bases}. This suggests that, when dealing with non self-adjoint hamiltonians, the assumption that their eigenstates form a basis could quite likely be false, while what might remain true is that these eigenstates produce some weaker form of resolution of the identity, as described by Assumption $\D$-pbw 3 in Section II. Also: even for quite simple systems, the metric operator can be unbounded, together with its inverse. Hence, for infinite-dimensional systems, more care is required than that usually used in the physical literature on the subject.

We should also mention that another non trivial output of this paper is that most of the results deduced in \cite{bagnewpb} under the very strong assumption that the  eigenstates of $N$ and $N^\dagger$ are indeed bases, still hold true even if they are simply $\D$-quasi bases.

\section*{Acknowledgements}

This work was partially supported by the University of Palermo.


\begin{thebibliography}{99}


\bibitem{petr} D. Krejcirik and P. Siegl,{\em On the metric operator for the imaginary cubic oscillator}, Phys. Rev. D, {\bf 86}, 121702(R) (2012)

\bibitem{circu}  J. Schindler, A. Li, M. C. Zheng, F. M. Ellis and T. Kottos, {\em Experimental study of active LRC circuits with PT symmetries}, Phys. Rev. A,  {\bf 84}, 040101, (2011)

\bibitem{ben}  C. Bender, {\em Making Sense of Non-Hermitian Hamiltonians}, Rep. Progr.  Phys., {\bf 70},  947-1018 (2007)

\bibitem{mosta} A. Mostafazadeh, {\em Pseudo-Hermitian representation of Quantum Mechanics}, Int. J. Geom. Methods Mod. Phys. {\bf 7}, 1191-1306 (2010)

\bibitem{znorev}
M. Znojil, {\em Three-Hilbert-space formulation of Quantum Mechanics}, SIGMA {\bf 5} 001 (2009)


\bibitem{zno} M. Znojil and M. Tater, {\em CPT-symmetric discrete square well},
Int. J. Theor. Phys., {\bf 50},  982- 990 (2011); M. Znojil, {\em Discrete quantum square well of the first kind},
Phys. Lett. A, {\bf 375},  2503-2509 (2011).

\bibitem{bagnewpb} F. Bagarello, {\em More mathematics for pseudo-bosons},  J. Math. Phys., {\bf 54}, 063512 (2013)

\bibitem{heil} C. Heil, {\em A basis theory primer: expanded edition}, Springer, New York, (2010)

\bibitem{ben2} C. M. Bender, H. F. Jones, {\em Interactions of Hermitian and non-Hermitian Hamiltonians}, J. Phys. A, {\bf 41}, 244006  (2008); Jun-Qing Li, Qian Li, Yan-Gang Miao, {\em Investigation of PT-symmetric Hamiltonian Systems from an Alternative Point of View}, Commun. Theor. Phys., {\bf 58}, 497
(2012); Jun-Qing Li, Yan-Gang Miao, Zhao Xue, {\em Algebraic method for pseudo-Hermitian Hamiltonians}, arXiv:1107.4972 [quant-ph]

\bibitem{bagpb2} F. Bagarello {\em Construction of pseudo-bosons systems},  J. Math. Phys., {\bf 51},  023531 (2010) (10pg)

\bibitem{kolfom} A. Kolmogorov and S. Fomine, {\em El\'ements de la th\'eorie des fonctions et de lanalyse fonctionelle}, Mir (1973)

\end{thebibliography}
\end{document}